\documentclass[
  a4paper,USenglish,cleveref,autoref,thm-restate,authorcolumns
]{lipics-v2021}

\pdfoutput=1
\hideLIPIcs
\nolinenumbers

\usepackage{csquotes}
\usepackage{booktabs}

\usepackage{local-models}
\usepackage{ammacros}

\makeatletter
\let\c@author\relax
\makeatother
\usepackage{ambib}
\usepackage{tikz}
\usepackage[mode=buildnew,subpreambles=true,sort=true]{standalone}

\crefname{step}{Step}{Steps}

\usepackage{xcolor}
\usepackage[draft]{fixme}
\fxusetheme{color}
\fxsetface{inline}{\small}

\FXRegisterAuthor{am}{anam}{\color{blue}AM}

\FXRegisterAuthor{tb}{antb}{\color{purple}TB}

\FXRegisterAuthor{ad}{anad}{\color{red}AD}

\FXRegisterAuthor{fc}{anfc}{\color{teal}FC}

\FXRegisterAuthor{me}{anme}{\color{olive}ME}

\FXRegisterAuthor{hl}{hli}{\color{violet}HL}

\title{Orientation does not help with 3-coloring a grid in online-LOCAL}

\author{Thomas Boudier}
  {Gran Sasso Science Institute, Italy}
  {thomas.boudier@gssi.it}
  {https://orcid.org/0009-0002-0119-9333}
  {}

\author{Filippo Casagrande}
  {Gran Sasso Science Institute, Italy}
  {filippo.casagrande@gssi.it}
  {https://orcid.org/0009-0009-6916-9289}
  {}

\author{Avinandan Das}
  {Aalto University, Finland \and \url{https://sites.google.com/view/avinandan/}}
  {avinandan.das@aalto.fi}
  {}
  {}

\author{Massimo Equi}
  {Aalto University, Finland \and \url{https://massimoequi.github.io/}}
  {massimo.equi@aalto.fi}
  {https://orcid.org/0000-0001-8609-0040}
  {}

\author{Henrik Lievonen}
  {Aalto University, Finland \and \url{https://henriklievonen.fi/}}
  {henrik.lievonen@aalto.fi}
  {https://orcid.org/0000-0002-1136-522X}
  {}

\author{Augusto Modanese}
  {Aalto University, Finland \and \url{https://augusto.modanese.net/}}
  {augusto.modanese@aalto.fi}
  {https://orcid.org/0000-0003-0518-8754}
  {}

\author{Ronja Stimpert}
  {Aalto University, Finland}
  {ronja.stimpert@aalto.fi}
  {https://orcid.org/0009-0000-5794-0034}
  {}

\authorrunning{T.~Boudier et al.}
\Copyright{Thomas Boudier, Filippo Casagrande, Avinandan Das, Massimo Equi,
  Henrik Lievonen, Augusto Modanese, and Ronja Stimpert}


\acknowledgements{This work was supported in part by the Research Council of
Finland, Grants 359104 and 363558, as well as the Quantum Doctoral Education
Pilot, the Ministry of Education and Culture, decision VN/3137/2024-OKM-4.}

\begin{document}

\maketitle

\begin{abstract}
  The online-LOCAL and SLOCAL models are extensions of the LOCAL model where
  nodes are processed in a sequential but potentially adversarial order.
  So far, the only problem we know of where the global memory of the
  online-LOCAL model has an advantage over SLOCAL is 3-coloring bipartite
  graphs.
  Recently, Chang et al.\ [PODC 2024] showed that even in grids, 3-coloring
  requires $\Omega(\log n)$ locality in deterministic online-LOCAL.
  This result was subsequently extended by Akbari et al.\ [STOC 2025] to also
  hold in randomized online-LOCAL.
  However, both proofs heavily rely on the assumption that the algorithm does
  not have access to the orientation of the underlying grid.
  In this paper, we show how to lift this requirement and obtain the same lower
  bound (against either model) even when the algorithm is explicitly given a
  globally consistent orientation of the grid.
\end{abstract}


\section{Introduction}

The \emph{\olcl} model \cite{akbari23_locality_icalp} bridges the fields of
online and distributed algorithms.
As in standard online algorithms, the model consists of a centralized
computational instance that must solve a task on a graph presented as a
(potentially adversarially constructed) stream.
The twist is that the algorithm maintains a \emph{local} view of the graph that
grows as more nodes are revealed:
With every node $v$ that arrives, the view of the algorithm is augmented by all
connections in the $T$-neighborhood of $v$ (for some function $T = T(n)$ where
$n$ is the final size of the graph). 
For every new node $v$ that is revealed, the algorithm must immediately commit
to a solution at $v$ using the view that has been revealed thus far.
Unlike standard questions in online algorithms, we disregard the space and time
complexities required to produce the output and are interested only in its
\emph{locality} $T$.

A fundamental problem in the \olcl model is \emph{$3$-coloring grids}.
This problem is important because it gives a separation between \olcl and the
weaker but related \slocal model \cite{ghaffari17_complexity_stoc}.
Namely, there is an $O(\log n)$ strategy for the problem in (deterministic)
\olcl (which works not only in grids but also in bipartite graphs), whereas in
\slocal it has complexity $n^{\Omega(1)}$ \cite{akbari23_locality_icalp}.

More recently, it was shown that $\Omega(\log n)$ locality is necessary to solve
this problem \cite{chang24_tight_podc}, even in the presence of randomness
\cite{akbari25_online_stoc}.
The approach for this lower bound, as given by \textcite{chang24_tight_podc},
is based on the idea of a \emph{potential} function that is defined in terms of the
$3$-coloring of the grid.
The argument consists of two main steps:
\begin{enumerate}
  \item Show there is an adversarial strategy for constructing rows that
  guarantees the following: 
  Somewhere along the row, there exists a pair of nodes $u$ and $v$ whose
  potential difference (in absolute value) is $\Omega(\log n)$.
  \item Using this strategy, create one such row $R_1$ and another (arbitrarily
  constructed) row $R_2$ of the same size.
  (See \cref{fig:chang_square}.)
  Consider the nodes $u'$ and $v'$ of $R_2$ that correspond to $u$ and $v$ and
  position the nodes $u$, $v$, $u'$, and $v'$ as corners of a square with
  minimal height.
  If necessary, rotate $R_2$ by 180 degrees so that the potential differences in
  both horizontal sides of the square have the same sign.
  Since a full walk around the square has zero potential, $\Omega(\log n)$ space
  between the rows is required in order to \enquote{build down} the difference
  between $u$ and $v$.
  The distance between the rows immediately translates into a lower bound for
  the locality $T$.
  \begin{figure}
    \centering
    \includegraphics[width=.85\textwidth]{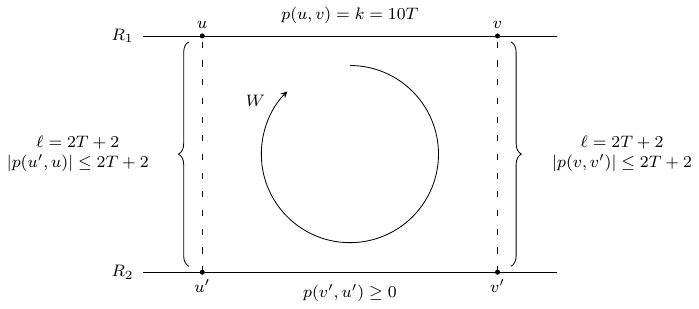}
    \caption{Illustration of the contradiction step in the proof of
    \textcite{chang24_tight_podc}.
    The first row $R_1$ is created with a potential difference $p(u,v) = 10T$
    between nodes $u$ and $v$.
    Parallel to this, we create $R_2$ and rotate it as necessary to ensure that we
    have nodes $u'$ and $v'$ parallel to $u$ and $v$ and such that $p(v',u') \ge
    0$.
    Finally, the two rows are revealed to be at distance $\ell = 2T+2$ from one
    another.
    The walk $W$ through $u$, $v$, $v'$, and $u'$ has potential $p(W) = p(u,v) +
    p(v,v') + p(v',u') + p(u',u) > 5T$, but the potential of any closed walk
    must be $0$ if the walk is properly 3-colored.}
    \label{fig:chang_square}
  \end{figure}
\end{enumerate}
Note that the second step in the proof requires a peculiar transformation, namely
rotating rows by 180 degrees, which potentially destroys the
orientation of the underlying grid.
The proof for the randomized case \cite{akbari25_online_stoc} exhibits the same
limitation.
Hence both proofs \emph{break down} if the input graph also contains the
orientation of the final grid.
\emph{Does this mean there is hope for an $o(\log n)$ strategy if we are given
the orientation of the grid?}

\paragraph*{Our Contribution}
In this paper, we answer this question in the negative for both the
deterministic and randomized settings.

\begin{theorem}
  \label{thm:main}
  In both the deterministic and randomized \olcl models, the complexity of
  $3$-coloring grids is $\Theta(\log n)$, even if the algorithm is given the
  orientation of the final grid.
\end{theorem}

With an orientation, we mean a globally consistent labeling of edges such that
every node knows the cardinal direction of each of its neighbors (i.e., north,
east, west, or south).
Since our contribution consists of enhancing the previous proofs so that they
also hold in this more general setting, it is necessary to present some key
ideas of those proofs first, for which we take a short detour at this juncture
before returning to the discussion of the actual proof.

\paragraph*{The Proof Strategy of \textcite{chang24_tight_podc}}

We now describe the proof of \textcite{chang24_tight_podc} in enough detail for
the unfamiliar reader to be able to follow the rest of the discussion in this
introduction.
Further technicalities needed for the proof in this paper are postponed to
\cref{sec:chang-lb}.

As already foreshadowed, proving a lower bound in \onlinelocal of $f(n)$ for
some function $f$ generally entails devising an adversarial strategy that
succeeds in \enquote{breaking} an algorithm with locality less than $f(n)$.
Hence, the proof of \textcite{chang24_tight_podc} (as well as that of
\textcite{akbari25_online_stoc} and ours) is concerned with describing such a
strategy.

The strategy of \textcite{chang24_tight_podc} bases on a \emph{potential
function} $p$ defined on walks.
The function $p$ counts the (signed) number of transitions between colors $1$
and $2$.
(See \cref{sec:chang-lb} for a precise definition.)
It also has some nice properties (\cref{lem:potential-properties}):
For a closed walk $W$, we have $p(W) = 0$, and the parity of $p(P)$ on a path
$P$ is uniquely determined by the colors of its endpoints and (the parity of)
its length.

As described above, the proof of \textcite{chang24_tight_podc} is structured
into two stages.
The first of these concerns devising a \emph{potential-boosting} strategy to
construct a row with (nearly) logarithmic potential increase $k = 10T$ (from
left to right).
Using this strategy, in the second stage, we obtain the contradiction proper by
constructing one such path $P_1$ as well as a parallel path $P_2$ that can be
transformed to have non-negative potential (again from left to right) assuming
the relative orientation of $P_1$ to $P_2$ is not known by the algorithm.
As already explained above, we can then position $P_1$ and $P_2$ adequately to
deduce $T = \Omega(\log n)$.

\paragraph*{Technical Overview}

To obtain our \cref{thm:main}, we need to rework the strategy just outlined in
\emph{several regards}:
\begin{enumerate}
  \item We argue that the potential-boosting strategy given by
  \textcite{chang24_tight_podc} can be used to obtain not only a logarithmic but
  \emph{quasilinear} potential difference along a row
  (\cref{sec:qlinear-boosting}).
  \item We demonstrate how the original logarithmic boosting strategy can be
  generalized so that we can construct not only rows but (approximately) lines
  in the plane of any desired slope $\theta \in [0,2\pi)$
  (\cref{sec:slope-boosting}).
  (Although we could also obtain quasilinear potential boosting in this extended
  case, we refrain from pursuing this extension in the present paper since it is
  not needed for the result.)
  \item Finally, we provide a completely novel approach to obtain the final
  contradiction that does not rely on transforming any object previously
  revealed in the construction (\cref{sec:contradiction}).
\end{enumerate}
Overall, compared to the simpler proofs of
\textcite{chang24_tight_podc,akbari25_online_stoc}, we must exploit the
\emph{full two-dimensional aspect} of the problem.

Next, we present the three different aspects just outlined in more detail. 
Along the way, we discuss the main difficulties to be overcome while sketching
an outline of the full argument.
From here on, we assume we are given a deterministic \onlinelocal algorithm with
locality $T = o(\log n)$ that we wish to prove to be incorrect for solving the
$3$-coloring problem.

\subparagraph{Step 1: quasilinear potential boosting.}\label{quasi-potential-boosting}
Let us inspect the potential-boosting strategy of \textcite{chang24_tight_podc}
more closely.
It consists of an iterative process that increases the guaranteed potential
difference by at least one unit at every step while at the expense of (roughly)
doubling the size of the construction.
Since there are only $n$ nodes at our disposal, it appears that the best we can
hope for is a $\Omega(\log n)$ potential difference; otherwise, we would exceed
our \enquote{budget} of nodes.

However, as it turns out, there is a more clever argument that we can apply:
Suppose we have constructed \emph{two} distinct rows $R_1$ and $R_2$ (not
necessarily using the aforementioned strategy) of roughly the same size whose
relative location is yet to be revealed.
Consider any two paths $P_1$ and $P_2$ of equal length on $R_1$ and $R_2$,
respectively, where $P_1$ and $P_2$ have the same orientation (e.g., both $P_1$
and $P_1$ go from west to east).
Let $p(P_i)$ denote the potential difference between the nodes at the start and
end of path $P_i$.
Then, the quantity $\abs{p(P_1) - p(P_2)}$ cannot be too large; otherwise we
could reveal $R_1$ and $R_2$ to be close to each other so that the endpoints of
the paths align and directly obtain a contradiction as in the second stage of
the proof of \citeauthor{chang24_tight_podc}.
To be precise, if we can reveal the respective endpoints of $P_1$ and $P_2$ to
be at distance $d$ from one another, then necessarily $\abs{p(P_1) - p(P_2)} \le
2d$ as the potential cannot drop by more than one unit at each node.
(In fact, as already shown in \textcite{chang24_tight_podc}, the potential drops
by at most one unit every \emph{three} nodes, so we even have a stricter upper
bound of $2d/3$.)

Imagine now that $R_1$ was constructed according to the logarithmic boosting
strategy and $P_1$ is a path guaranteeing $\abs{p(P_1)} \ge 10T$ (which is
possible since $T = o(\log n)$).
Because we can certainly reveal two separate rows at distance $d \le 3T$, we
\emph{immediately} get that \emph{any other path} $P_2$ of the same length as
$P_1$ in any other separately constructed row $R_2$ has $\abs{p(P_2)} \ge 4T$.
Assuming we needed roughly $2^{10T}$ nodes in $R_1$ to construct such a $P_1$,
it follows that any row of length $\ell \cdot 2^{10T}$ has potential difference
at least $4T \ell$.

\subparagraph{Step 2: boosting along arbitrary slopes.}
Although the logarithmic boosting strategy is defined to work on a row, it is
conceivable to try and apply it to any line $L$ on the plane having a slope
$\theta \in [0,2\pi)$.
Certainly, we cannot obtain a perfect match to any such line since we can only
position points on the lattice $\Z \times \Z$, but it is nevertheless easy to
see that we can pick a series of points in $\Z \times \Z$ that are at distance
from $L$ that is less than $1$.
Hence it seems plausible that one could just perform the construction using this
set of points instead, thus obtaining a difference in potential also for points
at marginal distance from $L$.

However, recall that the boosting strategy requires \emph{shifting} a fragment
of the construction relative to another so that the potential difference
increases.
Since $L$ may avoid points on $\Z \times \Z$ entirely, it is
simply \emph{not possible} to port the original potential boosting argument to
this case without substantial modifications.
In particular, the resulting object containing a path with high potential
difference \emph{necessarily} is a (true) two-dimensional construct with a
non-zero area.

This issue alone would not be very dramatic if we did not much care about the size
of the resulting object.
Unfortunately, the argument for the final contradiction (see further below) is
very sensitive to the height of this object.
In fact, we hope for an object guaranteeing a potential difference of $k$ and
having a height that does not exceed (roughly) $3k/2$.

To cope with these issues, we provide a careful and systematic construction that
minimizes the height of the resulting object.
More precisely, the object we construct is a \emph{parallelogram} where two of
its sides are vertical and of height $k+1$ and the other two sides have the
desired angle $\theta$ respective to the horizontal axis.
Along the way, we also handle several other minor technicalities that require
special care due to the two-dimensional nature of these objects.

\subparagraph{Step 3: obtaining a contradiction.}
Finally, we describe the idea for arriving at a contradiction.
The most intuitive way to grasp it is to imagine what kind of strategy an
algorithm would need to follow in order to not immediately fail against the
quasilinear boosting of Step 1.
For the sake of exposition, let us assume that we can achieve not only
quasilinear boosting but actually linear boosting.
This means that, along a row, the algorithm \emph{must} be doing at least one
transition from color $1$ to color $2$ every $r$ steps for some constant $r$
(that is not nullified by a subsequent transition back from $2$ to $1$).
Again, for the sake of simplicity, suppose that the algorithm does not perform
$2$ to $1$ transitions at all.
Hence, the potential is \emph{monotonically increasing} (with a slope of nearly
$1/r$) along the row.

Suppose we continue this reasoning until we have a row $R$ from $u$ to $v$ with
$n^{1/10}$ nodes (and thus $\Omega(n^{1/10})$ potential difference).
Here, we select the node $v$ at the far right and start revealing nodes along its
\emph{column} $C$.
By a similar argument as in Step 1, the algorithm also has to ensure a linear
potential difference along $C$ and, since $C$ has two directions, one of these
must be decreasing in potential.
If we continue revealing nodes in this direction, then eventually we reach a
node $w$ that has a potential difference of \emph{zero} to $u$.
This means we have found a \emph{diagonal} $D$ going from $u$ to $w$ that has
zero potential (see \cref{fig:l-shape}).

\begin{figure}
    \centering
    \includegraphics[width=.5\textwidth]{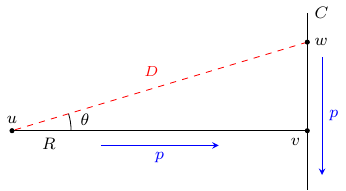}
    \caption{Illustration of constructing a diagonal $D$ with $p(D) = 0$.
    The blue arrows indicate the direction in which the potential is increasing.
    The potential increases (near) linearly from $u$ to $v$ along $R$, and so
    along $C$ we can find a node $w$ that has potential difference zero to $u$,
    thus defining $D$, which forms an angle of $\theta$ with $R$.}
    \label{fig:l-shape}
\end{figure}

With further observations, we can argue that $D$ must contain a path $P_0$ that
is not too long, say $2^{10T}$ nodes in total, that itself has potential zero.
Letting $\theta$ be the angle between $R$ and $D$, we use the construction of
Step 2 to obtain a parallel path $P_1$ with potential $k = 10T$.
Finally, we reveal $P_1$ to be at distance $2T + 2$ from $P_0$.
Having done so, we can argue as in the original contradiction proof of
\textcite{chang24_tight_podc} that the closed walk around $P_0$ and $P_1$ must
itself have potential zero, thus obtaining the desired contradiction.

As a final technicality, recall that in Step 2 it was mentioned that we desire
$P_1$ to be inside an object that has moderate height.
We can now observe where exactly this is needed:
Since the objects in Step 2 have some height $h$, the distance between $P_0$ and
$P_1$ is not only $2T + 2$ but actually this value plus potentially $h$ (if
$P_1$ is, say, at the very top of the constructed object).
Hence, since we can only control the potential along $P_0$ and $P_1$ and not
over the segments that connect the two, we certainly need $h = O(k)$ and,
moreover, we care about the actual constant.
Thankfully, it is a straightforward observation that the potential of a walk $W$
cannot be larger than $\abs{W}/3 + c$ (\cref{lem:pot-upper-bound}) for a
constant $c$.
This fact was not explicit in the work of \textcite{chang24_tight_podc} (and was
not needed for the results there), and so we give a short proof of it in
\cref{sec:chang-lb}.
Using this fact together with our Step 2 that guarantees $h = k+1$, we obtain
that the segments connecting $P_0$ and $P_1$ have length at most $h+2T+2$ and
thus contribute at most (roughly) $(h+2T)/3 +c \approx 4T$ in potential each,
and hence we still have a remaining potential of $k - 8T = 2T > 0$ for the
closed walk around $P_0$ and $P_1$.

\subparagraph{The randomized lower bound.}
The above strategy describes an adversary against \emph{deterministic}
\onlinelocal, which is allowed to make adaptive choices during the lower bound
construction.
For the randomized case, we must instead design a strategy for an
\emph{oblivious} adversary.
Following \textcite{akbari25_online_stoc}, we describe how to transform all of
the adaptive choices during the construction into non-adaptive ones by simply
\enquote{guessing} and boost the success probability by repetition.
In retrospect, this can be seen as a standard technique for porting lower bounds
from \detolcl to the randomized model.
Since this technique is very general and applies to all lower bounds for
\onlinelocal that the authors are aware of, there is reason to believe a
separation between the two models (which remains an interesting open problem)
will be very challenging to obtain.
In fact, in light of related results by \textcite{akbari25_online_stoc}, it is
plausible to believe one can obtain derandomization results for \onlinelocal in
a very wide range of settings.

\paragraph*{Outline}
The rest of the paper is structured as follows:
\Cref{sec:preliminaries} introduces main notation and the models at hand.
In \cref{sec:chang-lb} we review in full detail the core technical aspects of
\cite{chang24_tight_podc} that are needed for the proof.
Following that, in \cref{sec:det-lb,sec:rand-lb} we prove our result for
deterministic and randomized \onlinelocal, respectively.


\section{Preliminaries}
\label{sec:preliminaries}

We write $\N_+$ for the set of positive integers, $\N_0$ for $\N_+ \cup \{ 0
\}$, and $\Z$ for the set of all integers.
The set of the first $n \in \N_+$ positive integers is denoted by $[n]$.

An \emph{$(n\times n)$-oriented grid} $G$ is a directed graph defined over the vertex set $[n] \times [n]$, where $n \in \N_+$. The edge set of $G$ consists of consistently oriented edges defined as follows:
\begin{itemize}
    \item For all $1 \le i \le n$ and $1 \le j < n$, there is a directed edge from $(i,j)$ to $(i,j+1)$.
    \item For all $1 \le i < n$ and $1 \le j \le n$, there is a directed edge from $(i,j)$ to $(i+1,j)$.
\end{itemize}
The orientations of the edges are uniform and not arbitrary. Any graph that is isomorphic to $G$ via a bijection that preserves edge directions is also referred to as an \emph{oriented grid}.

We emphasize that while $G$ is defined as a directed graph, the edge
orientations serve as auxiliary information. In particular, graph-theoretic
notions such as distance are defined with respect to the underlying undirected
graph (i.e., by ignoring the directions of the edges).

With this in place, we can formally state the $3$-coloring problem considered in
this paper:

\begin{quote} \itshape
  Given an $(n\times n)$-oriented grid $G$, output a proper 3-coloring of $G$.
\end{quote}

Finally, let us define the two models in question.
For a vertex $v \in V(G)$ and $T \in \N_0$, we write $B(v, T)$ for the set of
vertices within (undirected) distance at most $T$ from $v$.

\subparagraph{\boldmath The \olcl model.}
The \olcl model was introduced by~\cite{akbari23_locality_icalp}. In this setting, the nodes of an $n$-vertex input graph $G$ are revealed one by one in an adversarial order $v_1, v_2, \ldots, v_n$. When the $i^{\text{th}}$ vertex $v_i$ is revealed, an \olcl algorithm with locality $T$ must assign an output label to $v_i$, based solely on the graph induced by the union of $T$-radius neighborhoods of the vertices revealed so far, that is, $G\left[\bigcup_{j \le i} B(v_j, T)\right]$, and on the sequence of previously revealed vertices $v_1, \ldots, v_i$ (along with their assigned labels).

\subparagraph{\boldmath The \rolcl model.}
The \ROlcl model~\cite{akbari25_online_stoc} is a randomized variant of the \olcl model. In this setting, the algorithm is additionally given access to a string of random bits. However, the adversary is non-adaptive—the input graph $G$ and the vertex reveal order $v_1, \ldots, v_n$ are fixed independently of the algorithm's random choices. As in the deterministic case, when vertex $v_i$ is revealed, the algorithm must assign it a label using the subgraph $G\left[\bigcup_{j \le i} B(v_j, T)\right]$, the order of revealed vertices so far, and its internal randomness.

\section{Relevant Notions from
\texorpdfstring{\textcite{chang24_tight_podc}}{Chang et al. [Cha+24]}}
\label{sec:chang-lb}

In this section, we recall notions and results from
\textcite{chang24_tight_podc} that are needed to establish our result.
The concept of central importance is that of a potential function defined on
paths of the grid, which itself is based on the 3-coloring assigned to the grid.
Fix a $3$-colorable graph $G=(V,E)$ and a proper coloring $c: V \to \{1,2,3\}$
of $G$.
The potential function $p$ is defined as follows:

\begin{definition}[Potential function]
  For any edge $\{u,v\} \in E$, we let
    \begin{equation*}
      p(u, v):=
      \begin{cases}
          c(u) - c(v) & \text{if $c(u) \neq 3$ and $c(v) \neq 3$,} \\
          0           & \text{otherwise.}
      \end{cases}
    \end{equation*}
    Given any walk $W$, we then let
    \begin{equation*}
      p(W) := \sum_{\{u,v\} \in W} p(u,v).
    \end{equation*}
\end{definition}

Note that the potential of any edge is in the set $\{-1, 0, 1\}$. 
In other words, we can see the potential function as a storage of potential over
a walk. 
Crossing the walk from color $2$ to color $1$ increases the potential by
$1$, and vice-versa.
As proven by \textcite{chang24_tight_podc}, the potential function satisfies
certain important properties:

\begin{lemma}[Properties of the potential function $p$]
  \label{lem:potential-properties}
  The function $p$ satisfies the following:
  \begin{enumerate}
    \item For any directed closed walk $W$ in a grid, $p(W) = 0$.
    \item For any path $P$ of length $\ell$ from $u$ to $v$, we have $p(P)
    \equiv i(u) + i(v) + \ell \pmod 2$, where $i(v)$ (resp., $i(u)$) is $1$ if
    $c(v) = 3$ (resp., $c(u) = 3$) or $0$ otherwise.
  \end{enumerate}
\end{lemma}

In addition, we also use the following property, which is not stated explicitly
in \cite{chang24_tight_podc} but is simple to prove.

\begin{lemma}[Upper bound on $p$.]
  \label{lem:pot-upper-bound}
  For any walk $W$, $p(W) \le \abs{W}/3 + O(1)$.
\end{lemma}

\begin{proof}
  Recall that, for every edge in $W$, the value of $p$ can only increase or
  decrease by one.
  We argue without restriction that, for any subwalk $W' = (u_1,u_2,u_3,u_4)$ of
  length $3$ where $p(u_1,u_2) = 1$, $\abs{p(W')} \le 1$.
  This is easy to see:
  By definition of $p$, $c(u_1) = 2$ and $c(u_2) = 1$.
  Then, since $c$ is a proper $3$-coloring, either $c(u_3) = 2$, in which case
  $\abs{p(W)} = \abs{p(u_3,u_4)} \le 1$, or $c(u_3) = 3$ and then $p(W) =
  p(u_1,u_2) = 1$.
\end{proof}

The main technical result of \textcite{chang24_tight_podc} that we build on is
the \emph{potential boosting lemma}, with which we can create a path with
logarithmic potential difference. 

\begin{lemma}[Logarithmic potential boosting]
  \label{lem:chang-potential-boosting}
  There is an adversarial strategy against $o(\log n)$-locality \detolcl
  that, given any $k \in O(\log n)$ and any desired length $N \in \N_+$ with $N
  < \sqrt{n}$, constructs a row of length $N$ containing a path $P$ with
  $\abs{p(P)} = k$.
\end{lemma}

Since we adapt the construction of \cref{lem:chang-potential-boosting}
(Lemma~3.6 in \cite{chang24_tight_podc}, arXiv version), we revise it in more detail.
The idea is to proceed inductively on $k$, where for $k = 0$ we start with a
single point of the grid.
To guarantee a path having potential difference $k$, it is useful to refer to the objects constructed on the inductive hypothesis as \emph{levels} of the construction.
Having built two rows $R_1$ and $R_2$ each at level $k-1$, we align $R_2$ to the right of $R_1$ along the same horizontal axis, choose the distance to place it from $R_1$, and then reveal all the nodes in between.  We now want to produce the path of level $k$. 
The relevant quantity here is the distance $\ell$ of the rightmost node of $R_1$
to the leftmost one of $R_2$, which we pick to be either $2T + 2$ or $2T + 3$.
Since $\ell > 2T + 1$, the views of $R_1$ and $R_2$ do not overlap (and hence it
is justified to handle them as separate objects).
To choose between the two values $2T + 2$ and $2T + 3$, consider the following:
suppose that in $R_i$, $i \in \{1,2\}$, we have a path $P_i$ from $u_i$ to $v_i$
that has potential $p(P_i) = k-1$
(note the sign of $p(P_i)$ is relevant here.)
Pick $\ell$ so that the value $p(P_3)$ of the path $P_3$ from $v_1$ (the rightmost node of $P_1$) to $u_2$ (the leftmost node of $P_2$) has
different parity from $k-1$.
(Here, we use the second item of \cref{lem:potential-properties}.)
This implies either we already have $\abs{p(P_3)} \ge k$ or then $\abs{p(P_3)}
\le k-2$ (since $\abs{p(P_3)} \neq k-1$), from which a simple calculation 
reveals that the path from $u_1$ to $v_2$ has potential difference at least $k$.


\section{\boldmath Lower Bound for \DetOlcl}
\label{sec:det-lb}

First, we prove that, without loss of generality, the potential increases everywhere in the north to south and west to east directions. We achieve that in \cref{lem:row-boosting} by using the adversarial strategy of \cref{lem:chang-potential-boosting}. Using this, we can deduce an adversarial strategy that creates a horizontal path $(a,b)$ and a vertical path $(c,b)$ meeting at node $b$, such that the potential of the walk $(a,\dots,b,\dots,c)$ is $0$. Using the first item of \cref{lem:potential-properties}, we then deduce that a path following the slope $A_0$ starting at $a$ and ending at $c$ must also have potential $0$. Finally, we show in \cref{lem:slope-boosting} that using a technique similar to \cite{chang24_tight_podc}, we can produce a slope $A_1$ with an orientation similar to $A_0$ but with too much potential increase to finish the parallelogram formed by $A_1$ and $A_0$ in accordance to \cref{lem:potential-properties}.

Throughout we assume the underlying graph $G$ is a \((\sqrt n\times \sqrt
n)\)-oriented grid and there is a \detolcl algorithm with locality $T = o(\log
n)$ that purportedly $3$-colors $G$.

\subsection{Quasilinear Boosting on a Row}
\label{sec:qlinear-boosting}

The first step is to reason about what occurs when two rows are brought next to
each other.

\begin{lemma}\label{lem:difference-potential}
  Let $L_1$, $L_2$ be two rows a distance of at least $2T+2$, such that they have not yet been revealed to be in the same component. Let $W_1$ be a walk on $L_1$ and $W_2$ be a walk on $L_2$, both of the same length. Then the potential difference $\abs{p(W_1)-p(W_2)}$ is at most $4T+4$.
\end{lemma}

\begin{proof}
  Suppose the algorithm discovered two walks of same length $W_1 = (u_1, v_1)$, $W_2=(u_2, v_2)$ such that $\abs{p(W_1)-p(W_2)} \ge 4T+5$. Without loss of generality, let $p(W_1)= K$ and $p(W_2)=K+ \eps (4T+5)$ with $\eps \in \{1, -1\}$. 
  The adversary then reveals that $W_1$ and $W_2$ are aligned at distance
  exactly $2T +2$. 
  By \cref{lem:potential-properties} the closed walk $C$ defined by $(u_1, v_1, v_2,
  u_2, u_1)$ must have potential $0$: 
  \[
    p(W_1) + p((v_1, v_2)) - p(W_2) + p(u_2, u_1) = 0,
  \]
  and thus 
  $
    p((v_1, v_2)) + p((u_2, u_1)) = \eps (4T+5) 
  $. 
  However, $\abs{p(u_2, u_1)} = \abs{p(v_1, v_2)} = 2T+2$, so $\abs{p((v_1,
  v_2))}$ and $\abs{p((u_2, u_1))}$ are at most $2T+2$, a contradiction.
\end{proof}

With the above observation in place, we can improve the strategy of
\cref{lem:chang-potential-boosting} to give us nearly linear potential
difference along a row.

\begin{lemma}\label{lem:row-boosting}
  For $T = o(\log n)$,
  there exists an adversarial strategy that ensures the existence of a value
  $c = 1/2^{O(T)}$ such that we can build rows with
  the following property: for any subpath $W'$ of the row with $\abs{W'} \ge
  2^{10T}$, we have $\abs{p(W')} \ge c \abs{W'}$.
\end{lemma}

\begin{proof}
  Using the same adversarial strategy as in \cref{lem:chang-potential-boosting},
  build a path $W$ with potential increase $9T$ and at most $2^{10T}$ length.
  This is possible for large enough values of $T$ (and $n$).
  Because of \cref{lem:difference-potential}, any parallel path $W'$ satisfying its conditions must have potential between $5T-4$ and $13T+4$.
  Moreover, if we split a row into $m$ consecutive segments of length $\abs{W}$, the potential of the whole row is between $m(5T-4)$ and $m(13T+4)$.
  Hence, we can force paths of length $N = \omega(2^{10T})$ such that any of its subpaths $W'$ of length at least $\abs{W} \le 2^{10T}$ has potential increase at least $c\abs{W'}$ with $c = 4T / \abs{W} \ge 4T/2^{10T}$.
\end{proof}

\subsection{Boosting on Arbitrary Slopes}
\label{sec:slope-boosting}

In this section, we show how one can apply the idea of
\cref{lem:chang-potential-boosting} along arbitrary lines on the plane, not just
horizontal and vertical rows and columns.
The key insight here is that the general idea of potential boosting relies only
on the fact that the parity of the potential is uniquely determined by the
distance between the vertices and their colors, not the actual path taken.
Hence we can make this path follow an arbitrary fixed slope that is determined
by two grid points.

That being said, generally the path cannot follow the sloped line exactly as it
must be aligned with the grid points.
To account for this, we allow the path to deviate from the target slope by a
small amount that we can control.
In particular, we restrict the path to fully lie between two parallel lines
whose vertical distance is bounded.
One way to interpret this is to consider the smallest parallelogram enclosing
the construction with two of its sides being vertical and two aligned with the
target slope.

\begin{lemma}[Slope boosting]
  \label{lem:slope-boosting}
  Given an oriented rectangular grid $G$ on $n$ nodes, a target slope $\theta
  \in [0,2\pi)$, and a target potential $k = o(\log n)$, there exists an
  adversarial strategy against any \detolcl algorithm with locality $T = o(\log
  n)$ such that the adversary constructs a path $P$ in $G$ with the following
  properties:
  \begin{enumerate}
    \item $P$ contains two points $u'$ and $v'$ such that the potential between
    them is exactly $k$, and
    \item $P$ is fully contained in a parallelogram with two of its sides being
    vertical with length $k+1$, the other two sides having angle $\theta$, and
    the total width being $2^k \cdot O(T)$.
  \end{enumerate}
\end{lemma}

We start by introducing some helpful terminology.
Consider the embedding $\tilde{G}$ of $G$ on the Euclidean plane that sends $v$
to the origin and places nodes of $G$ on the lattice $\Z \times \Z$ while
identifying the orientation of $G$ with the orientation of the plane and keeping
edges at unit length.
Recall that a line $L$ in $\tilde{G}$ is uniquely described by a point $v$ and a
slope $\theta$.
Here we only need to consider the case where $v \in V(G)$, so we cater our
definitions to it.
Given $i \in \N_0$, let $d(v,i) = d(v,i;\theta)$ be the unique (positive) real
such that $D(v,i) = (i,d(v,i)) \in L$.
We write $D_-(v,i) = D_-(v,i;\theta)$ for the node $v_- \in V(G)$ that has
coordinates $(i,\floor{d(v,i)})$ in $\tilde{G}$.
Meanwhile, we let $D_+(v,i) = D_+(v,i;\theta)$ denote the node $v_+ \in V(G)$
having coordinates $(i,\ceil{d(v,i)})$ in $\tilde{G}$ if $d(v,i)$ is not an
integer (i.e., $\floor{d(v,i)} \neq \ceil{d(v,i)}$), or $(i,d(v,i)+1)$
otherwise.
Note this means $D_-(v,i), D_+(v,i)$ are adjacent grid points along the same
column.
Intuitively, node $D_-(v,i)$ is on or below the line $L$ at horizontal distance
$i$ from $v$, and node $D_+(v,i)$ is above the line.

\begin{proof}
  Let $T = (1/100) \log n$.
  Without restriction, we consider only the case where $\theta \in [0,\pi/4]$.
  The construction for the general case is obtained by applying reflections
  along the horizontal and vertical axis or by swapping the two.

  As in \cref{lem:chang-potential-boosting}, our strategy is based on incrementally
  building self-similar objects while driving the potential difference up by at
  least one unit at every level.
  From one level to the next, the size of the construction should increase by
  only a multiplicative constant factor.
  Indeed, similar to \cref{lem:chang-potential-boosting}, the construction at
  level $j+1$ is obtained by creating two copies of it at level $j$ and
  adequately placing them relative to each other.
  
  The construction at level $0$ consists of a single node.
  Higher levels of the construction are shaped as the objects already mentioned
  in the statement of the lemma, namely parallelograms where two sides are
  vertical and of length equal to one plus the current level of the construction $j \le
  k$ and the two other sides have angle $\theta$ relative to the horizontal
  axis.  

  To simplify the description, we use the term \emph{$(j,\theta)$-parallelogram} to
  refer to such parallelograms.
  Given such a parallelogram $P$, we call the node of $P$ with minimal
  coordinates (in lexicographic order) the \emph{anchor node} of $P$.
  (Note node here refers to nodes of $G$ in the embedding $\tilde{G}$, not
  points on the plane.)
  We guarantee through the construction that every $(j,\theta)$-parallelogram
  has the same height $j+1$ and width $w_j$ and that these are both integral.
  These technical details are somewhat tedious but needed to ensure we reach the
  desired height for the construction.

  Let us describe the construction inductively (see also
  \cref{fig:slope-boosting}).
  To obtain a parallelogram at level $j+1$ from two $(j,\theta)$-parallelograms
  $P_1$ and $P_2$ at level $j$, we proceed as follows:
  \begin{enumerate}
    \item Let $u_1$ and $u_2$ be the anchor nodes of $P_1$ and $P_2$,
    respectively.
    \item Set $i = w_j+2T+2$.
    Note that, since we assumed $\theta \in [0,\pi/4]$, this implies both
    $D_-(u_1,i)$ and $D_+(u_1,i)$ are at distance greater than $2T+1$ but still
    at most $O(T)$ from $P_1$.
    \item \label[step]{step:arbitrary-slopes-adaptive-step}
    Place $P_2$ relative to $P_1$ by placing $u_2$ on top of $D_-(u_1,i)$
    or $D_+(u_1,i)$, whichever of the two choices guarantees that the maximal
    potential difference increases by at least one unit.
    (This choice is made using the same reasoning as in the setting of
    \cref{lem:chang-potential-boosting}.)
    \item \label[step]{step:arbitrary-slopes-new-parallelogram}
    Recall $P_1$ and $P_2$ both have height $j+1$. 
    Hence, if we could place $P_2$ on $D(u_1,i)$ (the ideal point to follow the
    slope $\theta$), then we would be actually able to fit both $P_1$ and $P_2$
    in a $(j,\theta)$-parallelogram perfectly.
    Since both points $D_-(u_1,i)$ and $D_+(u_1,i)$ are at most unit distance
    from $D(u_1,i)$, it follows that there is a $(j+1,\theta)$-parallelogram $P$
    containing both $P_1$ and \emph{both the placement choices} for $P_2$.
    (In particular, this implies the shape of the $(j+1,\theta)$-parallelogram
    at the next level of the construction is uniquely defined.)
    \item Once $P$ is fixed, reveal all nodes inside it that have not been
    revealed previously, thus reaching the next level of the construction.
  \end{enumerate}
  Note the choice of placement for $P_1$ and $P_2$ ensures the relative position
  of the two is kept hidden from the algorithm until both $P_1$ and $P_2$ have
  been completely labeled.
  The argument for the increase in potential difference is the same as in
  \cref{lem:chang-potential-boosting}.
  
  \begin{figure}
    \centering
    \includegraphics[width=.9\textwidth]{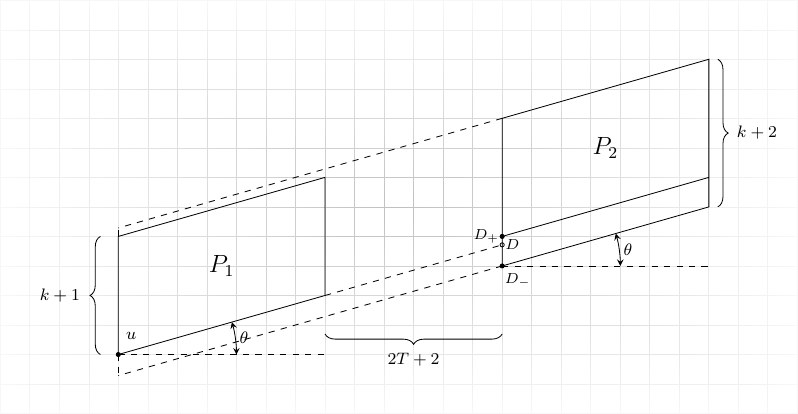}
    \caption{Illustration of the potential-boosting procedure along a slope
    $\theta$.
    Here we are placing two parallelograms $P_1$ and $P_2$ of height $k+1$
    relative to one another.
    The node $u$ is the anchor node of $P_1$.
    No matter which choice we make between placing the anchor node of $P_2$ at
    $D_-$ or $D_+$, the resulting parallelogram has height $k+2$ (and in fact
    encompasses both choices).}
    \label{fig:slope-boosting}
\end{figure}

  Due to the relative placement of $P_1$ and $P_2$, it is evident that we can
  guarantee the height and width guarantees mentioned previously.
  One technical issue concerns the construction at level $1$:
  There it is needed that the parallelograms have height $2$ as otherwise it
  might be that there is no path of $\tilde{G}$ that is fully contained within
  it.
  This is required so that the potential difference is defined between any two
  nodes in the same parallelogram (and also the sole reason why we only obtain
  height $k+1$ and not $k$ in the statement).

  Finally, we verify the size of the construction allows us to reach the desired
  level $j = k$:
  Clearly, we have $w_{j+1} = 2(w_j + T + 1)$.
  Since $w_0 = 0$, by solving the recursion, we obtain $w_k = 2(2^k-1)(T+1)$.
  As the final parallelograms have height $k+1$, we need thus at most $O(Tk)
  \cdot 2^k = n^{o(1)}$ nodes for the construction.
\end{proof}

\subsection{Obtaining a Contradiction}
\label{sec:contradiction}

In this section, we conclude the lower bound proof for the deterministic case.
We first sketch the main ideas intuitively.
We begin by constructing a path $W$ from some node $u$ to some node $v$ whose
shape resembles the letter \enquote{L}, that is, a row followed by a column. 
Taking advantage of the techniques already employed in
\cref{lem:difference-potential,lem:row-boosting}, we can guarantee that the
potential of this path is zero (see \cref{lem:L-path-construction} below). 
At this point, the shortest path $D$ from $v$ to $u$ proceeds in a diagonal
manner. 
Consider the walk combining $W$ with $D$. 
Then its potential must be zero since it is a closed walk that start from node
$u$, reaches $v$, and comes back to $u$.
Since the potential difference is unique, if we start from $u$ and follow any walk, then the potential reaches zero both at $v$ and at $u$.

We can use this fact to our advantage to obtain a contradiction. The idea is to
first observe that there must be two nodes $v_1$ and $u_1$ along the diagonal path
$D$ at which the potential assumes the same value up to a constant
(\cref{lem:diag-constant-potential}). Then, we construct another diagonal path
$R$ where the absolute value of the potential increases by $10T$
(\cref{lem:adverserial-strategy}). 
Finally we reveal that $R$ is close to $v_1$ and $u_1$. 

Now consider a path that starts from $v_1$, joins path $R$, follows it, and then
returns to $u_1$. 
The contradiction comes from the fact that the distance from the nodes $v_1$ and
$u_1$ and path $R$ is not enough to compensate for the growth of the potential that we
enforced, and thus now there is a closed walk starting and ending at $v_1$ and
passing through $R$ that does not have zero potential.

\begin{definition}[L-path]\label{def:L-path}
  A path $W=(w_1,w_2,\ldots,w_j)$ is an L-path if there exists a node $w_i\in W$, where $i<j$, such that all nodes of subpath $(w_1,w_2,\ldots,w_i)$ belong to the same row, and all nodes of subpath $(w_i,w_{i+1},\ldots,w_j)$ belong to the same column.
\end{definition}

\begin{definition}[diagonal-path; see also \cite{bresenham1965}]\label{def:diagonal-path}
Let $u=(x_0,y_0)$ and $v=(x_1,y_1)$ be two nodes of the grid 
with $x_0 < x_1$, and let
$m = (y_1-y_0)/(x_1-x_0)$
	be the slope\footnote{By \emph{slope}, we interchangeably mean the ratio $m$ or the angle $\theta = \arctan(m)$.} of the Euclidean line segment joining the nodes $u$ and $v$. Without loss of generality, assume that $\abs{m}\leq 1$. The \emph{diagonal-path} from $u$ to $v$ is
    \[
      P = \left\{ (x,y)\in \mathbb{Z}^2 : 
      \text{$x_0 \leq x \leq x_1$ and $y = y_0 + \floor{m(x-x_0)}$}
      \right\}.
    \]	
\end{definition}

\begin{lemma}[L-path construction]\label{lem:L-path-construction}
  There exists an adversarial strategy to create an L-path $W$ from node $u$ to node $v$ such that:
  \begin{itemize}
    \item the row subpath of $W$ has length at least $2^{100T}$,
    \item the row subpath of $W$ is longer than the column subpath of $W$,
    \item $\abs{W} \le 2^{O(T)}$,
    \item $p(W)=0$,
    \item the revealed nodes are only the ones in $W$.
  \end{itemize}
\end{lemma}
\begin{proof}  

  Our argument is independent of the actual orientation of the edges of the grid and only cares about the direction in which the potential is increasing or decreasing. For this reason, we can assume that we are always forcing the algorithm to increase or decrease the potential in one specific direction. This is because, if this is not the case, we could go through the same construction by virtually flipping the orientation. For example, if we want the potential to increase from west to east but the algorithm is making it decrease, in our construction we can consider west to be east and west to be east. Thus, w.l.o.g. in the following argument, we consider the potential to be always increasing from west to east and decreasing from south to north.

  Acting as the adversary, we start from a node $w_1$ and we reveal nodes to the algorithm belonging to the same row of $w_1$, going neighbor by neighbor always from west to east. Let the path revealed this way be $W_1=(w_1,w_2,\ldots,w_i)$. Thanks to \cref{lem:row-boosting}, the algorithm is forced to nearly linearly increase the potential along this row. Now consider the node $w_i$. We reveal nodes to the algorithm that belong to the same column of $w_i$, going neighbor by neighbor always from south to north.  Let the path revealed this way be $W_2=(w_i,w_{i+1},\ldots,w_j)$. Applying the same reasoning as in the proofs of \cref{lem:difference-potential} and \cref{lem:row-boosting}, the algorithm linearly decreases the potential along this column. Notice that we can choose $W_2$ to be long enough so that $p(W_1)+p(W_2)=0$. Let $\abs{p(W_1)}=c_1(T)\abs{U_1}$ and $\abs{p(W_2)}=c_2(T)\abs{W_2}$, where $c_1$ and $c_2$ are functions defined as in the statement of \cref{lem:row-boosting}. Let $c(T)=\frac{c_1(T)}{c_2(T)}$ be such that $\abs{W_2}=c(T)\abs{W_1}$. From \cref{lem:pot-upper-bound} we know that $c_1(T)\le \frac 13$, and \cref{lem:row-boosting} guarantees that $c_2(T)\ge \frac{1}{2^{O(T)}}$, thus $c(T)\le \frac{2^{O(T)}}{3}$. This implies that, if we choose $\abs{W_1}=O(2^{100T})$ so that we satisfy the lemma statement, we always have enough room to grow $W_2$ so that we satisfy condition $p(W_1)+p(W_2)=0$, and this is because $\abs{W_2}=c(T)\abs{W_1}\le  \frac{2^{O(T)}}{3}\abs{W_1}$.  Moreover, we can always swap the roles of $W_1$ and $W_2$ after their construction, and thus also guarantee that $\abs{W_1} \ge \abs{W_2}$.  Finally, we remark that $\abs{W_1}+\abs{W_2} \le 2^{O(T)}\cdot 2^{100T} = 2^{O(T)}$. We conclude by observing that $W=(w_1,w_2,\ldots,w_i,w_{i+1},\ldots,w_j)$ is an $L$-path that satisfies the statement of the lemma by taking $u=w_1$ and $v=w_j$.
\end{proof}

The adversary reveals an L-path $W$ from node $u$ to node $v$ following the strategy in~\cref{lem:L-path-construction} such that one of the arms of the L-path has length at least $2^{100T}$ and forcing the algorithm to color it in such a way that $p(W)=0$. It then reveals to the algorithm a \emph{diagonal-path} $D_{u,v}$ from node $u$ to node $v$.
In order to prove a certain property about $D_{u,v}$, we need the following two
simple analytical observations.

\begin{lemma}[Discrete intermediate value theorem]
    \label{lem:ivt}
    Let $b, k \in \mathbb{Z}_+$, and let $f : [0,b]_{\mathbb{Z}} \to \mathbb{Z}$ be a function satisfying
	    $\abs{f(x+1) - f(x)} \leq k$ for all  $x \in [0,b-1]_{\mathbb{Z}}$
    with boundary conditions $f(0) \ge 0$ and $f(b) \le 0$.
    Then there exists an integer $x \in [0,b]_{\mathbb{Z}}$ such that
    $
        \abs{f(x)} \le k .
    $
\end{lemma}
\begin{proof}
	We may assume that $f(0) > k$ and $f(b) < -k$ as otherwise we could pick either $x = 0$ or $x = b$.
	Let $a < b$ be the largest integer such that $f(a) > k$.
	Now $f(a + 1)$ is either in range $[-k, k]$ or it is $< -k$.
	In the former case we are done and can pick $x = a+1$.
	In the latter case $\abs{f(a+1) - f(a)} \ge 2k$, which contradicts our assumption.
\end{proof}

\begin{lemma}[Discrete mean value theorem]
    \label{lem:mvt}
    Let $b, k \in \mathbb{Z}_+$, and let $f : [0,b]_{\mathbb{Z}} \to \mathbb{Z}$ be a function satisfying
    $\abs{f(x+1) - f(x)} \leq k$ for all  $x \in [0,b-1]_{\mathbb{Z}}$,
    with boundary conditions $f(0) = f(b) = 0$.  
    Fix an integer $\ell$ with $0 < \ell < \sqrt{b}$, and define
    \[
    g(x) =  f(x+\ell) - f(x), \ x \in [0, b-\ell]_{\mathbb{Z}}.
    \]
    Then there exists some $x \in [0, b-\ell]$ such that $\abs{g(x)} \le 2k$.
\end{lemma}

\begin{proof}
If $g(0) = 0$, we are done.
W.l.o.g. assume $g(0) > 0$ as otherwise we consider the function $-f$.
We show there exists $x \in [0, b-\ell]_{\mathbb{Z}}$ such that $g(x) \le 0$; by \cref{lem:ivt} this then implies the claim as $g$ varies by at most $2k$ between consecutive indices.

Assume for contradiction that $g(x) > 0$ for all $x \in [0, b-\ell]_{\mathbb{Z}}$.
Indeed, we may even assume that $g(x) > 2k$ everywhere as otherwise we could pick $x$ where this does not hold.
Consider now the function $f$ under this assumption.
The first few values of $f$ are $f(0) = 0$, $f(\ell) > 2k$, $f(2\ell) > 4k$, and
so on since $g(x) > 2k$ everywhere, which forces $f$ to increase.
At the same time, $f$ can decrease under the assumption that the difference
between consecutive indices of $f$ is at most $k$.
Hence, we have $f(\ell^2) > 2kl$, $f(\ell^2+1) > 2k\ell-k$, $f(\ell^2+2) >
2k\ell-2k$, and so on, until finally $f(\ell^2+\ell-1) > 2k\ell - k(\ell-1) > 0$
and $f(\ell^2+\ell) > 2k^2+2k$.
Thus, for every $x \ge \ell^2$, the function $f$ never reaches a non-positive
value, contradicting $f(b) = 0$.
\end{proof}

With these technicalities in place, we now prove that we can find two nodes in
$D_{u,v}$ with \emph{constant} potential difference between them.

\begin{lemma}\label{lem:diag-constant-potential}
For every positive integer $i < 2^{50T}$, there exist nodes
	$u_1=(a,b)$ and $v_1=(a',b')$ in the diagonal-path $D_{u,v}$ with $\abs{a-a'} = i$ such that the diagonal
subpath from $u_1$ to $v_1$ has constant potential:
\[
	\abs{p\left(D_{u,v}[u_1,v_1]\right)} \leq 2.
\]
\end{lemma}

\begin{proof}
	Let $u=(\alpha,\beta)$ and $v=(\alpha',\beta')$ and let $W_{\text{row}}$ denote the row sub-path of $W$ and set $B \coloneqq \abs{W_{\text{row}}}$. By the construction of the L-path (~\cref{lem:L-path-construction}), we have that the angle $\theta$ between the Euclidean line $\ell_{u,v}$ through $u$ and $v$ and the $x$-axis lies in $[0,\pi/4]$ and $\abs{W_{\text{row}}}\geq 2^{100T}$.
	
	For each $x \in [0,B]_{\mathbb{Z}}$, let $w_x=(a_x,b_x)$ denote the unique node on $D_{u,v}$ such that $\abs{a_x-\alpha}=x$, and define the function $f:[0,B]\rightarrow \mathbb{Z}$ such that $f(x) = p(D_{u,v}[u,w_x])$.

By construction, $f(0)=0$, and by \cref{lem:L-path-construction},
$
    f(B) = p\left(D_{u,v}[u,v]\right) = p\left(Q[u,v]\right) = 0.
$

Next, consider two nodes $w_x=(a_x,b_x)$ and $w_{x+1}=(a_{x+1},b_{x+1})$ in $D_{u,v}$.  
	Observe that $\abs{a_{x+1}-a_x}=1$ and $\dist(w_x,w_{x+1}) \leq 2$ (which follows from the construction of $D_{u,v}$ which closely follows $\ell_{u,v}$ and the fact that $\theta\in [0,\pi/4]$). Therefore, the value of $f$ on the two nodes differs by at most $2$. More specifically,
    $\abs{f(x+1)-f(x)} \leq 2$
    for all  $x \in [0,B-1]_{\mathbb{Z}}$.

Now fix $i<2^{50T}$ and set $\ell \coloneqq i$.  
Since $B \geq 2^{100T}$, it follows that $0<\ell<\sqrt{B}$.  
Define
    $g(x) \coloneqq f(x+\ell)-f(x)$ for $x \in [0,B-\ell]_{\mathbb{Z}}$.
Applying the discrete mean value theorem (\cref{lem:mvt}) with $k=2$ yields an index $x^\star\in[0,B-\ell]_{\mathbb{Z}}$ such that $\abs{g(x^\star)}\leq 2$.

Finally, let
    $u_1 \coloneqq w_{x^\star} = (a_{x^\star},b_{x^\star})$
    and
    $v_1 \coloneqq w_{x^\star+\ell} = (a_{x^\star+\ell},b_{x^\star+\ell})$.
Then $\abs{a_{x^\star+\ell}-a_{x^\star}} = \ell = i$, and
$
    p\left(D_{u,v}[u_1,v_1]\right) = f(x^\star+\ell)-f(x^\star) = g(x^\star)
$,
which implies $\abs{p(D_{u,v}[u_1,v_1])} \leq 2$. 
\end{proof}

Finally, we describe the adversarial strategy that concludes the proof.

\begin{lemma}\label{lem:adverserial-strategy}
	There exists an adversarial strategy to reveal  a path $R$ starting from some node $v_2 = (a_{v_2},b_{v_2})$ and ending at some node $u_2 = (a_{u_2},b_{u_2})$ such that 
	\begin{enumerate}
		\item\label{potential-parallelogram} $p(R) = 10T$ and the path $R$ is enclosed in a \emph{$(10T+1,\theta)$-parallelogram} of width $2^{10T}\cdot O(T)$  where $\theta$ is the angle made by the line joining $u$ and $v$ with the x-axis.
		\item\label{constant-potential-diagonal} There exists a diagonal-subpath $D_{u,v}[u_1,v_1]$ from node $u_1 = (a_{u_1},b_{u_1})$ to node $v_1 = (a_{v_1},b_{v_1})$ such that $\abs{a_{u_1}-a_{v_1}} = \abs{a_{u_2}-a_{v_2}}$ and $\abs{p(D_{u,v}[u_1,v_1])}\leq 2$.
		\item\label{distance-up-lb} $2T+2\leq \dist(D_{u,v}[u_1,v_1], R)\leq 12T+7$
	\end{enumerate}
Here, the distance between \(R\) and \(D_{u,v}[u_1,v_1]\) is defined as
\[
    \dist\left(R, D_{u,v}[u_1,v_1]\right)
    = \min \left\{ 
      \dist(x,y) \mid x \in V(R), y \in V\left(D_{u,v}[u_1,v_1]\right) 
    \right\}.
\]
\end{lemma}

\begin{proof}
	By construction, the diagonal path $D_{u,v}$ can be enclosed in a $(4,\theta)$-parallelogram. The adversary reveals a path $R$ following the strategy in~\cref{lem:slope-boosting} with $p(R) = 10T$ and the path being enclosed in a $(10T+1,\theta)$-parallelogram of horizontal span $2^{10T}\cdot O(T)$ such that the vertical distance between the two parallelograms is $2T+2$ (The exact positioning of the revealed path will be fixed soon). Revealing such a path and forcing the quasilinear boosting on the path is possible as $T= o(\log n)$ and such a path is at a distance at least $2T+2$ from the diagonal-path $D_{u,v}$ as well as the L-path $W$, hence ensuring that the relative positions of the revealed paths remain indistinguishable to the algorithm. This proves~\cref{potential-parallelogram}. 

Since $R$ is contained in a parallelogram of width $2^{10T}\cdot O(T)$, we obtain
$
    \abs{a_{u_2}-a_{v_2}} \leq 2^{10T}\cdot O(T)
$. 
By \cref{lem:diag-constant-potential}, there exist nodes 
$u_1=(a_{u_1},b_{u_1})$ and $v_1=(a_{v_1},b_{v_1})$
with $\abs{a_{u_1}-a_{v_1}}=\abs{a_{u_2}-a_{v_2}}$ such that the diagonal subpath $D_{u,v}[u_1,v_1]$ satisfies 
$
	\abs{p\left(D_{u,v}[u_1,v_1]\right)}\leq 2
$,
establishing~\cref{constant-potential-diagonal}.
 The adversary positions the parallelogram of $D_{u,v}[u_1,v_1]$  keeping $u_1,v_2$ and $v_1,u_2$ aligned on common vertical axes while maintaining the vertical distance of $2T+2$ between the two parallelograms. The dimensions of the parallelograms imply
$
    2T+2 \leq \dist(D_{u,v}[u_1,v_1],R) \leq 12T+7
$, 
as required in~\cref{distance-up-lb}.  
\end{proof}

We now have all the ingredients to conclude the deterministic lower bound.
The adversary reveals the paths  $P_{v_2,u_1}$ and $P_{v_1,u_2}$,  which are the
shortest paths connecting  the endpoints of $D_{u,v}[u_1,v_1]$ and $R$. 
Each of these length at least $2T+2$.
Consider the closed walk
\[
    W = D_{u,v}[u_1,v_1] \circ P_{v_1,u_2} \circ R \circ P_{v_2,u_1}.
\]
By \cref{lem:pot-upper-bound}, the potential of each of the connecting paths is
bounded, that is, both $\abs{p(P_{v_1,u_2})}$ and $\abs{p(P_{v_2,u_1})}$ are at
most $\tfrac{12T+7}{3} + c$ for some constant $c$.
Hence, if we choose $n$ large enough, then
$
    \abs{p(W)} \geq \abs{10T - 2\cdot \tfrac{12T+7}{3} - 2c} > 0
$, 
contradicting \cref{lem:potential-properties} and concluding the proof. 


\section{\boldmath Lower Bound for \ROlcl}
\label{sec:rand-lb}

In this section, we lift our lower bound against \detolcl to \emph{randomized}
\onlinelocal.
Recall that we must work with an \emph{oblivious} instead of an adaptive
adversary.
The idea is similar to the proof in \textcite{akbari25_online_stoc}:
Whenever our adversary makes an adaptive choice while executing the processes
described throughout \cref{sec:det-lb}, we instead blindly \emph{guess} one of
the choices available.
The probability of guessing all the correct choices needed to cause the
algorithm to fail is indeed small, but if we have a reasonable number of choices
and alternatives, then through repetition we can argue that our adversary
succeeds with at least constant probability, thus \enquote{breaking} the \rolcl
algorithm.

During this process, the key quantity to keep in mind is how many nodes from
our \emph{budget} are consumed by each choice we make (since the entire
construction entails revealing nodes to the algorithm) as well as what the
number of repetitions implies about how large this budget needs to be.
As an example, suppose that in every attempt to prove the algorithm wrong we
need to make $a$ many choices from $b$ many equally good alternatives and
consume $c$ many nodes in the process.
The success probability of one such attempt is $1/b^a$, and so we wish to repeat
this process $b^a$ many times to obtain a success probability larger than $1-1/e
> 1/2$.
In turn, this means we must have $b^a c \le n$ since otherwise we would have
exceeded the total number $n$ of nodes in the graph.

In what follows, we systematically consider all the choices and the number of
alternatives for the choices made during the process of \cref{sec:det-lb}.
We stress that this is the only adaptation that needs to be made and all of the
other details are identical to the deterministic case.

\subparagraph{Quasilinear potential boosting (\cref{sec:qlinear-boosting}).}
This is the simplest part to argue since \citeauthor{akbari25_online_stoc}
already gave a generalization of the potential boosting strategy to the
randomized case, where it succeeds with high probability and uses at most
$n^{o(1)}$ many nodes \cite[Lemma~8.6]{akbari25_online_stoc}.
The only detail we need to consider is in \cref{lem:difference-potential} where
we potentially reveal two of these paths; however, that certainly does not
change the budget usage of $n^{o(1)}$ since it is done only once.

\subparagraph{Boosting on arbitrary slopes (\cref{sec:slope-boosting}).}
When executing the procedure described in the proof of
\cref{lem:slope-boosting}, once $\theta$ and $k$ are fixed, the only choice to
be made is in \cref{step:arbitrary-slopes-adaptive-step} (i.e., choosing the
relative placement of the two parallelograms between two fixed options).
Note that in \cref{step:arbitrary-slopes-new-parallelogram} we already have
taken care to ensure that the subsequent stages of the construction can be made
oblivious to this choice (i.e., their shapes and positioning are the same
however the placement is made).

An important detail is that we insist that our adversary succeeds in \emph{every
single one} of the $k$ steps in the construction.
This is because, as already discussed in the introduction, it is imperative that
the height of the parallelograms is kept under control and, in particular, we
wish to keep it at $k+1$.
The naive approach would be to repeat the construction as-is multiple
times---however, this is not feasible since there are in total roughly $2^k$
choices in total to be made up to the $k$-th level of the construction.
Since every choice fails with probability $1/2$, this would require $2^{2^k}$
repetitions, which potentially blows up our budget since we the best upper bound
on $k$ we have is $k = o(\log n)$.

Instead, we adopt the strategy used in \cite[Lemma~8.6]{akbari25_online_stoc}.
Namely, we repeat the construction constantly many times \emph{at every level}
so that at every level our success probability is at least $1/2$.
Hence the overall construction at level $k$ succeeds with $1/2$ probability, its
size remains $2^{O(k)} = n^{o(1)}$, and we need only $2^k = n^{o(1)}$
repetitions to ensure success with high probability.
Moreover, we can do the repetitions at every level in a smart way \emph{that
does not increase the height} of the resulting object.
To achieve this, we align the different copies of the construction at the
previous level $i$ so they all fit in the same $(i+1,\theta)$-parallelogram.
(Recall this is a detail that we guarantee in the proof of
\cref{lem:slope-boosting}.)
Hence we preserve the height $k+1$ for the $k$-th level of the construction
after this modification.

\subparagraph{Obtaining a contradiction (\cref{sec:contradiction}).}

The construction of \cref{sec:contradiction} requires multiple choices.
The first of these is the orientation of the L-path in
\cref{lem:L-path-construction}, though this is simple as there are only
constantly many choices available.
Since $T$ is known, once this choice is fixed, then the lengths of the two paths
$W_1$ and $W_2$ in the same lemma can be set non-adaptively, where the choice of
the lengths here comes from a space of $2^{O(T)} = n^{o(1)}$ alternatives.
Having done so, we have nodes $u$ and $v$ defining a diagonal path $D_{u,v}$ and
need to determine nodes $u_1$ and $v_1$ as in
\cref{lem:diag-constant-potential}, for which again there are at most
$\abs{D_{u,v}} = 2^{O(T)} = n^{o(1)}$ choices each.
Finally, we must construct the path $R$ of \cref{lem:adverserial-strategy}, for
which we invoke the reasoning already developed above.
In total, we are compounding a \emph{constant} number of choices over $n^{o(1)}$
alternatives in this step, and hence can obtain a \emph{constant} success
probability for the adversary without exceeding the budget of $n$ nodes.

\printbibliography

@inproceedings{chang24_tight_podc,
	author = {Chang, Yi{-}Jun and Mishra, Gopinath and Nguyen, Hung Thuan and Yang, Mingyang and Yeh, Yu{-}Cheng},
	bibsource = {dblp computer science bibliography, https://dblp.org},
	booktitle = {{Proceedings of the 43rd {ACM} Symposium on Principles of Distributed Computing, {PODC} 2024, Nantes, France, June 17-21, 2024}},
	doi = {10.1145/3662158.3662794},
	editor = {Gelles, Ran and Olivetti, Dennis and Kuznetsov, Petr},
	pages = {106–116},
	publisher = {{ACM}},
	title = {{A Tight Lower Bound for 3-Coloring Grids in the Online-LOCAL Model}},
	url = {https://doi.org/10.1145/3662158.3662794},
	year = {2024}
}

@inproceedings{akbari25_online_stoc,
	author = {Akbari, Amirreza and Coiteux{-}Roy, Xavier and D'Amore, Francesco and Gall, François Le and Lievonen, Henrik and Melnyk, Darya and Modanese, Augusto and Pai, Shreyas and Renou, Marc{-}Olivier and Rozhon, Václav and Suomela, Jukka},
	bibsource = {dblp computer science bibliography, https://dblp.org},
	booktitle = {{Proceedings of the 57th Annual {ACM} Symposium on Theory of Computing, {STOC} 2025, Prague, Czechia, June 23-27, 2025}},
	doi = {10.1145/3717823.3718211},
	editor = {Koucký, Michal and Bansal, Nikhil},
	pages = {1295–1306},
	publisher = {{ACM}},
	title = {{Online Locality Meets Distributed Quantum Computing}},
	url = {https://doi.org/10.1145/3717823.3718211},
	year = {2025}
}

@inproceedings{akbari23_locality_icalp,
	author = {Akbari, Amirreza and Eslami, Navid and Lievonen, Henrik and Melnyk, Darya and Särkijärvi, Joona and Suomela, Jukka},
	bibsource = {dblp computer science bibliography, https://dblp.org},
	booktitle = {{50th International Colloquium on Automata, Languages, and Programming, {ICALP} 2023, July 10-14, 2023, Paderborn, Germany}},
	doi = {10.4230/LIPICS.ICALP.2023.10},
	editor = {Etessami, Kousha and Feige, Uriel and Puppis, Gabriele},
	pages = {10:1–10:20},
	publisher = {Schloss Dagstuhl - Leibniz-Zentrum für Informatik},
	series = {{LIPIcs}},
	title = {{Locality in Online, Dynamic, Sequential, and Distributed Graph Algorithms}},
	url = {https://doi.org/10.4230/LIPIcs.ICALP.2023.10},
	volume = {261},
	year = {2023}
}

@inproceedings{ghaffari17_complexity_stoc,
	author = {Ghaffari, Mohsen and Kuhn, Fabian and Maus, Yannic},
	bibsource = {dblp computer science bibliography, https://dblp.org},
	booktitle = {{Proceedings of the 49th Annual {ACM} {SIGACT} Symposium on Theory of Computing, {STOC} 2017, Montreal, QC, Canada, June 19-23, 2017}},
	doi = {10.1145/3055399.3055471},
	editor = {Hatami, Hamed and McKenzie, Pierre and King, Valerie},
	pages = {784–797},
	publisher = {{ACM}},
	title = {{On the complexity of local distributed graph problems}},
	url = {https://doi.org/10.1145/3055399.3055471},
	year = {2017}
}

@article{bresenham1965,
  author  = {Bresenham, Jack E.},
  title   = {Algorithm for computer control of a digital plotter},
  journal = {IBM Systems Journal},
  volume  = {4},
  number  = {1},
  pages   = {25--30},
  year    = {1965},
  doi     = {10.1147/sj.41.0025}
}

\appendix

 \end{document}